\documentclass[onecolumn,twoside]{IEEEtran}
\IEEEoverridecommandlockouts

\usepackage{ifpdf}
\usepackage{cite}
\usepackage[cmex10]{amsmath}
\usepackage{enumerate}
\usepackage{amsthm}
\usepackage{latexsym}
\usepackage{amssymb}
\usepackage{xspace}
\usepackage{amscd}
\usepackage{enumerate}
\usepackage{graphicx}
\usepackage{epic}

\newtheorem{theorem}{Theorem}
\newtheorem{lemma}[theorem]{Lemma}
\newtheorem{proposition}[theorem]{Proposition}
\newtheorem{definition}[theorem]{Definition}
\newtheorem{corollary}[theorem]{Corollary}
\newtheorem{remark}[theorem]{Remark}

\newtheorem{example}[theorem]{Example}

\newcommand{\N}{{\mathbb N}}

\newcommand{\Z}{{\mathbb Z}}
\newcommand{\lcm}{\mathrm{lcm}}
\newcommand{\F}{\mathbb F}
\newcommand{\Le}{\mathbb L}
\newcommand{\K}{\mathbb K}
\newcommand{\R}{\mathcal R}
\newcommand{\tq}{\;\mid\;}

\begin{document}

\title{Cyclic and BCH Codes whose Minimum Distance Equals their Maximum BCH bound
\thanks{This work was partially supported by MINECO (Ministerio de Econom\'{\i}a
y Competitividad), (Fondo Europeo de Desarrollo Regional)
project MTM2012-35240, Programa Hispano Brasile\~{n}o de Cooperaci\'{o}n Universitaria PHB2012-0135, and Fundaci\'{o}n S\'{e}neca of Murcia. The second author has been supported by Departamento Administrativo de Ciencia, Tecnolog\'{\i}a e Innovaci\'on de la Rep\'ublica de Colombia}}

\author{\IEEEauthorblockN{Jos\'e Joaqu\'{i}n Bernal\IEEEauthorrefmark{1},
Diana H. Bueno-Carre\~no\IEEEauthorrefmark{2} and
Juan Jacobo Sim\'on\IEEEauthorrefmark{1}. \\
\IEEEauthorblockA{\IEEEauthorrefmark{1}Departamento de Matem\'{a}ticas\\
Universidad de Murcia,
30100 Murcia, Spain.\\ Email: \{josejoaquin.bernal, jsimon\}@um.es} \\
\IEEEauthorblockA{\IEEEauthorrefmark{2}Departamento de Ciencias Naturales y Matem\'{a}ticas\\
Pontificia Universidad Javeriana, 
 Cali, Colombia\\
 Email: dhbueno@javerianacali.edu.co}
}}

\maketitle

%% Enter the first author's name and address:
%\centerline{\scshape Jos\'{e} Joaqu\'{\i}n Bernal, Juan Jacobo Sim\'{o}n}
%\medskip
%{\footnotesize
%% please put the address of the first author
 %\centerline{Departamento de Matem\'{a}ticas, Universidad de Murcia, Spain}
%%    \centerline{Other lines}
%%    \centerline{ Springfield, MO 65801-2604, USA}
%} % Do not forget to end the {\footnotesize by the sign }
%
%\medskip
%
%\centerline{\scshape  and Diana H. Bueno-Carre\~no }
%\medskip
%{\footnotesize
 %% please put the address of the second  and third author
 %\centerline{ Departamento de Ciencias Naturales y Matem\'aticas,}
   %\centerline{Pontificia Universidad Javeriana-Cali, Colombia}
  %
   %%\centerline{Departamento de Matem\'{a}ticas, Universidad de Murcia, Spain}
%}
%
%\bigskip
%
%% The name of the associate editor will be entered by an editorial staff
%% "Communicated by the associate editor name" is not needed for special issue.
 %\centerline{(Communicated by the associate editor name)}

%The abstract of your paper
\begin{abstract}
In this paper we study the family of cyclic codes such that its minimum distance reaches the maximum of its BCH bounds. We also show a way to construct cyclic codes with that property by means of computations of some divisors of a polynomial of the form $x^n-1$. We apply our results to the study of those BCH codes $C$, with designed distance $\delta$, that have minimum distance $d(C)=\delta$. Finally, we present some examples of new binary BCH codes satisfying that condition. To do this, we make use of two related tools: the discrete Fourier transform and the notion of apparent distance of a code, originally defined for multivariate abelian codes.
\end{abstract}

%The title of your section 1

\section{Introduction}
The computation of the minimum distance of a cyclic code, or a lower bound for it, is one of the main problems on abelian codes (see, for example, \cite{Ch,S,VLW}). The oldest lower bound for the minimum distance of a cyclic code is the BCH bound \cite[p. 151]{HP}. The study of this bound and its generalizations is a classical topic which includes the study of the very well-known family of BCH codes. In particular, an interesting problem is to determine when the maximum of the BCH bounds of a given cyclic code equals its minimum distance (see \cite{Charpin,S}). This is our interest.

In  this paper we deal with three problems related to the study of the BCH bound. The first one is how to give necessary and sufficient conditions for a cyclic code to insure that the maximum of its  BCH bounds equals its minimum distance. The second problem is how to construct such cyclic codes. Our third problem is related to construction techniques of BCH codes for which its designed distance, its maximum BCH bound and its minimum distance coincide. 

To solve our first problem, we make use of two related tools: the discrete Fourier transform and the notion of apparent distance of a code, originally defined for multivariate abelian codes in \cite{C}. These tools and the notation needed are given in Section 2. In Section 3, we characterize those cyclic codes for which its minimum distance reaches the maximum of its BCH bounds (problem 1). Then we study how to construct cyclic codes with that property by means of computations of divisors of a polynomial of the form
  $x^n-1$ (problem 2). Section 4 is devoted to solve our third problem. We apply our results to the study of those BCH codes $C$, with designed distance $\delta$, that have minimum distance $d(C)=\delta$ (see \cite[Section 9.2]{S}). In this paper, some examples of construction techniques and examples of new binary BCH codes whose minimum distance equals its designed distance are presented. We point out that all computations were done by using the GAP4r7 program \cite{GAP} with the cooperation of Alexander Konovalov.   The authors are indebted to him.

\section{Notation and preliminaries}\label{preliminares}

We will use standard terminology from coding theory (see for example \cite[Chapter 7]{S} or \cite[Section 2]{Charpin}). We denote by $q$ a power of the prime number $p$ and by $\F_q$ the field of $q$ elements. Let $n$ be a positive integer which is coprime to $q$. We denote by $R_n$ the set of $n$-th roots of unity and by $U_n$ the set of \textit{primitive} $n$-th  roots of unity. 

We denote by $\F_q[x]$ the ring of polynomials with coefficients in $\F_q$. For any $g=g(x)\in \F_q[x]$ we denote by $\deg(g)$ its degree, by $supp(g)$ its support and by $\omega(g)=\left|supp(g)\right|$ its weight. For any positive integer $n$, we consider the quotient ring $\F_q[x]/(x^n-1)$ which will be denoted by $\F_q(n)$. As usual, we identify the elements $g\in \F_q(n)$ with polynomials; so we may take $g\in \F_q(n)$ and then write $g\in \F_q[x]$ (where $\deg(g)<n$). For any $f\in \F_q[x]$ we denote by $\overline{f}$ its image under the canonical projection onto $\F_q(n)$. 

As in \cite{VLW}, a cyclic code $C$ of length $n$ in the alphabet $\F_q$ will be identified with the corresponding ideal in $\F_q(n)$ (up to permutation equivalence). Then, by a cyclic code we mean an ideal of $\F_q(n)$. It is well known that if $\gcd(n,q)=1$ then the quotient ring $\F_q(n)$ is semisimple and then every cyclic code has a unique monic generator polynomial \cite[Theorem 7.1]{S} and a unique idempotent generator \cite[Theorem 8.1]{S}. We always assume that $\gcd(n,q)=1$.

We denote by $\Z_n$ the integers modulo $n$ and we identify any class in $\Z_n$ with its canonical representative. It is well-known that every cyclic code $C$ in  $\F_q(n)$ is totally determined by its set of zeros (or its root set), which is defined as $Z(C)=\left\{\alpha \in R_n \tq c(\alpha)=0,\;\text{ for all }\; c\in C\right\}$; thus, for any polynomial $f\in \F_q(n)$, we have that $f\in C$ if and only if $f(\alpha)=0$ for all $\alpha\in Z(C)$. Fixed $\alpha \in U_n$, we denote the defining set of $C$ with respect to $\alpha$ as $D_\alpha (C)=\left\{i\in \Z_n\tq \alpha^i\in Z(C)\right\}$ (see \cite[p. 199]{S}). It is well-known that, when $\gcd(n,q)=1$, defining sets are partitioned in $q$-cyclotomic cosets modulo $n$ \cite[p. 104]{S}, which are defined as follows: given any element $a\in \Z_n$, the $q$-cyclotomic coset of $a$ modulo $n$ is the set $C_q(a)=\{a, qa,\dots,q^{n_a-1}a\}(\mod n)$, where $n_a$ is the smallest positive integer such that $q^{n_a}a\equiv a\mod n$. We recall that the notions of set of zeros and defining set are also applied to polynomials in $\F_q(n)$ in the obvious way.

For any code $C$, we denote its minimum distance by $d(C)$. The BCH bound states that for any cyclic code in $\F_q(n)$ that has a string of $\delta-1$ consecutive powers of some $\alpha \in U_n$ as zeros, the minimum distance of the code is at least $\delta$ \cite[Theorem 7.8]{S}. In terms of defining sets, if there is a string of $\delta -1$ consecutive integers modulo $n$ in $D_\alpha(C)$, for some $\alpha \in U_n$, then $d(C)\geq \delta$. Note that different roots of unity may yield different defining sets and consequently different lower bounds. For any cyclic code $C$ the maximum of its BCH bounds will be denoted by $\Delta(C)$. Sometimes it is called \textit{the} BCH (lower) bound of the code (see \cite[p. 22]{C} and \cite[p. 984]{Charpin}). 

%As each BCH bound of a cyclic code depends on the elections of a primitive $n$-th root of unity, it is necessary to consider all possible elections of them, that may give us different defining sets. 

The following Remark shows that in order to compute the maximum $\Delta(C)$ we do not need to consider all the elements in $U_n$. This fact will be used later.

\begin{remark}\label{diferentescotas}
Let $C_q(a_1),\dots, C_q(a_h)$ be the $q$-cyclotomic cosets modulo $n$ and fix a complete set of representatives $\{a_1, \dots , a_h\}$. Suppose we have chosen $\alpha\in U_n$ to get a defining set $D_\alpha(C)$. We want to identify the elements $\beta \in U_n$ satisfying that $D_\beta(C)\neq D_\alpha(C)$. Then, $\beta$ must satisfy the equality $\beta^{a_iq^j}=\alpha$ for some representative $a_i$ with $\gcd(n,a_i)=1$ and $j\in \Z$. In this case $D_\beta(C)=a_i\cdot D_\alpha(C)$, where the multiplication has the obvious meaning. We define
\begin{equation}\label{A}
 A(n)=\{a_i \tq \gcd(a_i,n)=1\}.
\end{equation}
It is easy to see that $O_n(q)=|C_q(a_i)|$ for any $a_i\in A(n)$. In addition, since $D_{\beta^{a_i}}(C)=D_{\beta^{a_iq^j}}(C)=D_{\beta^{a_iq^{j'}}}(C)$ ($j,j'\in \Z$), we conclude that we have to consider at most $\frac{\phi(n)}{O_n(q)}$ distinct defining sets or elements in $U_n$ to get $\Delta(C)$. 

For example, set $n=41$ and $q=2$. The $2$-cyclotomic cosets are $C_2(0),\,C_2(1)$ and $C_2(3)$. So $A(41)=\{1,3\}$. Fixed $\alpha\in U_{41}$, let $C$ be the cyclic code with defining set $D_\alpha(C)= C_2(1)$. Some BCH bounds for $C$ with respect to $\alpha$ are $\delta_1=3$ by considering  $\{1,2\}\subset D_\alpha(C)$, and $\delta_2=4$ by considering $\{8,9,10\}\subset D_\alpha(C)$. Now we also have to consider $D_\beta(C)=3\cdot D_\alpha(C)= C_2(3)$ and compute the corresponding BCH bounds. We find $\delta_3=6$ by considering $\{11,12,13,14,15\}\subset D_\beta(C)$.  In this case, $\Delta(C)=6$. It is worth to mention that in the binary and ternary cases for $n\leq 70$ we have that $\frac{\phi(n)}{O_n(q)}\leq 6$ and for $n\leq 90$ we have that $\frac{\phi(n)}{O_n(q)}\leq 8$.
\end{remark}

A cyclic code $C$ in $\F_q(n)$, with generator polynomial $g(x)$, is a BCH code of designed distance $\delta$ if there exists $\alpha \in U_n$ and $b\in \{0,\dots,n-1\}$ such that $g(x)$ is the polynomial with the lowest degree over $\F_q$ such that $\left\{\alpha^{b+j}\tq j=0,\dots,\delta-2\right\}\subseteq Z(C)$ (see \cite[p. 202]{S}). Equivalently, $C$ is a BCH code if for any cyclotomic coset $Q\subseteq D_\alpha(C)$ we have that $Q\cap \left\{b+j\tq j=0,\dots,\delta-2\right\}\neq \emptyset$. As it is known, this implies that $C$ is the cyclic code with highest dimension such that its set of zeros satisfies the inclusion mentioned above. We denote such a code by $B_q(\alpha,\delta,b)$.  The Bose distance of a BCH code $C=B_q(\alpha,\delta,b)$ is defined as the largest $\delta'$ such that $C=B_q(\alpha',\delta',b')$, for some $b'\in \{0,\dots,n-1\}$ and some $\alpha'\in U_n$. We note that for a BCH code it may happen that its Bose distance is less than $\Delta(B_q(\alpha,\delta,b))$, as we shall see in the next example. 

Let $\Le|\F_q$ be an extension field. For any element $a\in \Le$ we denote by $\mathrm{min}_q(a)$ the minimal polynomial of $a$ in $\F_q[x]$. In the case $q=2$ we only write $\mathrm{min}(a)$.

\begin{example}\label{distancia de bose mayor que dist BCH}\rm{
Set $q=2$, $n=21$ and fix $\alpha\in U_{21}$ such that $\mathrm{min}(\alpha)=x^6+x^5+x^4+x^2+1$. Let $C=B_2(\alpha,4,6)$ be the BCH code generated by $\lcm\{\mathrm{min}(\alpha),\mathrm{min}(\alpha^{3}), \mathrm{min}(\alpha^{7})\}$. Consider the 2-cyclotomic cosets modulo $21$, $C_2(0)=\{0\}$, $C_{2}(1)= \{1,2,4,8,11,16\}$, $C_2(3)= \{3,6,12\}$, $C_2(5)=\{5,10,13,17,19,20\}$, $C_2(7)=\{7,14\}$ and $C_2(9)=\{9,15,18\}$.  One may check that the defining set of the code $C$ with respect to $\alpha$ is $D_\alpha(C)= C_2(1)\cup C_2(3)\cup C_2(7)= C_2(6)\cup C_2(7)\cup C_2(8)$. In this case $A(21)=\{1,5\}$ so we also have to consider the element $\beta\in U_{21}$ such that $\beta^5=\alpha$. Then  $D_\beta(C)=5 \cdot D_\alpha(C)=C_2(5)\cup C_2(7)\cup C_2(9)$. One may see that the Bose distance is $\delta=4$, given by considering $\{6,7,8\}\subset D_\alpha(C)$ and $\{13,14,15\}\subset D_\beta(C)$. However $\Delta(C)=5$, because $\{1,2,3,4\}\subset D_\alpha(C)$ and $\{17,18,19,20\}\subset D_\beta(C)$. But $\{1,2,3,4\}\subset C_2(1)\cup C_2(3)$ and $\{17,18,19,20\}\subset C_2(5)\cup C_2(9)$, so that $C$ cannot be
  a BCH code of designed distance $\delta=5$. Hence the Bose distance is less than the maximum of all possible BCH bounds (or simply the BCH bound, $\Delta(C)$).
}\end{example}

Let $\Le|\F_q$ be an extension field such that $U_n \subseteq \Le$ and fix $\alpha\in U_n$. The \textit{ (discrete) Fourier transform} of a polynomial $f\in \F_q(n)$ with respect to $\alpha$ (also called Mattson-Solomon polynomial), that we denote by $\varphi_{\alpha,f}$ is defined as $$\varphi_{\alpha,f}(x)=\sum_{j=0}^{n-1} f(\alpha^j)x^j.$$
Clearly, $\varphi_{\alpha,f}\in \Le(n)$; moreover, the function Fourier transform may be viewed as an isomorphism of algebras $\varphi_{\alpha}:\Le(n)\longrightarrow (\Le^{n},\star)$, where the multiplication ``$\star$'' in $\Le^n$ is defined coordinatewise (see \cite[Section 2.2]{C} or \cite[\S~8.6]{S}). Then we may see $\varphi_{\alpha,f}$ as a vector in $\Le^n$ or as a polynomial in $\Le(n)$. The  inverse of the Fourier transform is given by $\varphi^{-1}_{\alpha,g}(x)=\frac{1}{n}\sum_{i=0}^{n-1}g(\alpha^{-i})x^i$, where $g\in \Le(n)$ (see for example \cite{C,Charpin,S}). For any  $i \in \{0,\dots,n-1\}$ we denote $\varphi_{\alpha,f}[i]=f(\alpha^i)$, the coefficient (or coordinate) corresponding to $x^i$.

\begin{remark}\label{basicos de la DFT}
 For any  $\alpha\in U_n$, $f\in \F_q(n)$ and $g\in \Le(n)$ we have that:
\begin{enumerate}
\item $supp\left(\varphi_{\alpha,f}\right)=\left\{ i\in \{0,\dots,n-1\}\tq f\left(\alpha^i\right)\neq 0\right\}$ and hence $\Z_n\setminus supp\left(\varphi_{\alpha,f}\right)=D_\alpha(f)$, the defining set of $f$.
\item Since $f=\varphi^{-1}_{\alpha,\varphi_{\alpha,f}}(x)$ then $supp(f)=\left\{ i\in \{0,\dots,n-1\}\tq \varphi_{\alpha,f}\left(\alpha^{-i}\right)\neq 0\right\}$, so that $|supp(f)|=n-|Z\left (\varphi_{\alpha,f}\right)|$.
\item $\varphi^{-1}_{\alpha,g}\in \F_q(n)$ if and only if  $\left(g\left(\alpha^j\right)\right)^q=g\left(\alpha^j\right)$ for any $j\in \{0,\dots,n-1\}$. %Moreover, in this case, $supp(g)$ is a union of $q$-cyclotomic cosets.
\item $\varphi^{-1}_{\alpha,g}\in \F_q(n)$ if and only if $\varphi^{-1}_{\beta,g}\in \F_q(n)$ for all $\beta\in U_n$.
\end{enumerate}
The first two assertions come directly from the definition of the discrete Fourier transform together with the fact that it is an isomorphism. The third one comes directly from the well-known property that an element $a\in \Le$ satisfies that $a\in \F_q$ if and only if $a^q=a$. Finally to see the last assertion observe that if we take another primitive root of unity $\beta\neq\alpha$ the coefficients of $\varphi^{-1}_{\beta,g}$ are obtained by permuting those of $\varphi^{-1}_{\alpha,g}$.
\end{remark}

The following lemma, related with the discrete Fourier transform, will play an important role later.
\begin{lemma}\label{TFDenFq}
Let $g\in \Le(n)$. If $\varphi^{-1}_{\alpha,g}\in \F_q(n)$ for any $\alpha\in U_n$ then $supp(g)$ is a union of cyclotomic cosets. If $g$ is an idempotent in $(\Le^n,\star)$ the converse holds; that is, if $supp(g)$ is union of $q$-cyclotomic cosets then $\varphi^{-1}_{\alpha,g}\in \F_q(n)$.
\end{lemma}
\begin{proof}
First, suppose that $\varphi^{-1}_{\alpha,g}\in \F_q(n)$. Observe that for any $f(x)\in \F_q(n)$, $\varphi_{\beta,f}^q (x)=\sum_{j=0}^{n-1}\left(f\left(\beta^j\right)\right)^q x^j=\sum_{j=0}^{n-1}\left(f\left(\beta^q\right)\right)^j x^j=\varphi_{\beta^q,f}(x)$. So the defining set of $\varphi^{-1}_{\alpha,g}$ is a union of cyclotomic cosets. Since $supp(g)=\Z_n\setminus D_\alpha(\varphi^{-1}_{\alpha,g})$ we are done.  

We first note that any idempotent in $(\Le^n,\star)$ verifies that its coordinates (or coefficients) are only 1 or 0. Now, suppose that $g\in(\Le^n,\star)$ is an idempotent and $supp(g)$ is a union of $q$-cyclotomic cosets. Then there exists an idempotent $e\in \F_q(n)$ such that $D_\alpha(e)=\Z_n\setminus supp(g)$; in fact, $e$ is the idempotent generator of the code over $\F_q$ with defining set $\Z_n\setminus supp(g)$ with respect to $\alpha$. Since $e$ is an idempotent in $\F_q(n)$ we have that $\varphi_{\alpha,e}$ is an idempotent in $(\Le^n,\star)$ and also $supp(\varphi_{\alpha,e})=supp(g)$. Then $\varphi_{\alpha,e}=g$ and hence $\varphi^{-1}_{\alpha,g}\in \F_q(n)$. \end{proof}

Let us recall some definitions in \cite[Chapter 3]{C} related to the computation of the BCH bound. The context of these definitions is the study of multivariate polynomials.  We only need the univariate polynomials version. 

\begin{definition}\label{aparentepol}
 Let $\Le$ be a field. For any element $g\in \Le(n)$ we define the \textit{apparent distance} of $g$, that we denote by $d^\ast(g)$, as follows 
\begin{enumerate}
 \item If $g=0$ then $d^\ast(0)=0$.
\item If $g\neq 0$ then
	\[d^\ast (g)=\max\left\{n-\deg\left(\overline{x^hg}\right)\tq 0\leq h\leq n-1\right\}.\]
\end{enumerate}
\end{definition}

It is easy to see that one may compute the apparent distance of a polynomial $0 \neq g\in \Le(n)$ as follows. Suppose that $g=\sum a_ix^i$. If we associate to the polynomial its coefficient vector $M(g)=\left(a_0,\dots, a_{n-1}\right)$ then the apparent distance $d^\ast(g)$ is the length of the biggest chain of consecutive zeros (modulo $n$) in $M(g)$ plus $1$.

\begin{example}\rm{
 Let $f=1+x+x^4\in \F_2(5)$. Compute $ \overline{x^0f}=1+x+x^4$,\; $\overline{xf}=1+x+x^2$;\;  $\overline{x^2f}=x+x^2+x^3$;\; $ \overline{x^3f}=x^2+x^3+x^4$;\; $\overline{x^4f}=1+x^3+x^4$. 
Then $d^\ast(f)=5-\deg(\overline{xf})=3$.

 If we take} $M(f)=(1\;1\;0\;0\;1)$ then $d^\ast(f)=2+1=3$.
\end{example}

Let $f\in \Le(n)$. It is clear that the polynomials $f$ and $\overline {x^h f}$ have the same set of zeros (or root set). Hence, $\deg\left(\overline{x^hf}\right) \geq | D_\alpha (f)|$, for any $\alpha\in U_n$, where $D_\alpha(f)$ denotes the defining set of $f$. Therefore $d^\ast(f)\leq n- | D_\alpha (f)|$ for any $\alpha\in U_n$. 

 Now, by the definition of the inverse Fourier transform (see Remark~\ref{basicos de la DFT}), we have that 
\begin{equation}\label{peso y dominio de def}
\omega(f)=n-| D_\alpha (\varphi_{\alpha,f})|.
\end{equation}

 Hence, 
\begin{equation}\label{dist aparente imagen menor que peso}
 d^\ast(\varphi_{\alpha,f})\leq n-| D_\alpha (\varphi_{\alpha,f})|= \omega(f),\;\;\text{ for all }\;\;f\in \F_q(n)\;\;\text{ and }\;\;\alpha\in U_n.
\end{equation}

This implies that the minimum of the apparent distances of the images of the nonzero codewords of a cyclic code is a lower bound for its minimum distance. Camion's definition of apparent distance of an abelian code comes from these ideas. In our case, we present that definition as follows.

\begin{definition}\label{distancia aparente}
  Let $C$ be a cyclic code in $\F_q(n)$ and consider $\alpha \in U_n$. The \textit{ apparent distance of $C$ with respect to $\alpha$} is $d^\ast_{\alpha}(C)=\min_{c \in C,\,c\neq 0}\{d^\ast(\varphi_{\alpha,c}) \}$ and \textit{the apparent distance of $C$} is $$d^\ast(C)=\max_{\alpha\in U_n} \{d^\ast_{\alpha}(C)\}.$$  We also define the set of \textit{ optimal roots} of $C$ as 
  \[\R(C)=\left\{\beta \in U_n \tq d_{\beta}^\ast(C)=d^\ast(C)\right\}.\]
\end{definition}

For the paragraph prior Definition \ref{distancia aparente} we have that $d^\ast(C)\leq d(C)$ for any cyclic code $C$.  In \cite[p. 22]{C} Camion shows that for any cyclic code $C$ the equality $d^*(C)=\Delta(C)$ holds.

Note that the value $d^\ast (\varphi_{\alpha,c})$ depends on the support of $\varphi_{\alpha,c}$; that is, it depends on the \textit{distribution} of the zeros of $c$ with respect to $\alpha$; so, the minimum $d^\ast_\alpha(C)$ depends on the \textit{distribution} of $D_\alpha(C)$. Hence, in order to compute the maximum $d^\ast(C)$ we need to look at the different defining sets of $C$, for each $\alpha\in U_n$. As we have seen in Remark \ref{diferentescotas}, if we fix $\alpha\in U_n$ and $\{a_1,\dots, a_h\}$, a complete set of representatives of the $q$-cyclotomic cosets modulo $n$, to consider the different defining sets of $C$ we only need to consider the roots $\beta\in U_n$ such that $\beta^{a_i}=\alpha$ for some $a_i$ coprime with $n$. Then, for any $\alpha\in U_n$ we define the set 
\begin{equation}\label{Ra}
 \mathcal{R}_\alpha=\{\beta \in U_n \tq \beta^{a}=\alpha, \, \, a \in A(n)\}.
\end{equation}
where $A(n)$ was defined in (\ref{A}).
 
Therefore, in practice, to compute the apparent distance of a cyclic code $C$ in $\F_q(n)$ it is enough to fix $\alpha \in U_n$ and compute $d^\ast(C)=\max \{d^\ast_{\beta}(C) \tq \beta \in \mathcal{R}_\alpha\}.$ 

 Let $e,g\in C$ be the idempotent generator and the generator polynomial of $C$, respectively. If $f,h\in \F_q(n)$ then $supp\left(\varphi_{\beta,fh}\right) \subseteq supp\left(\varphi_{\beta,f}\right)$ because $\varphi_{\beta,fh} = \varphi_{\beta,f}\star \varphi_{\beta,h}$, and then, for any $c\in C$, and any $\beta \in U_n$, we have that $supp\left(\varphi_{\beta,c}\right)\subseteq supp\left(\varphi_{\beta,g}\right)=supp\left(\varphi_{\beta,e}\right)$; so that, $d^\ast\left(\varphi_{\beta,g}\right)=d^\ast\left(\varphi_{\beta,e}\right)\leq d^\ast\left(\varphi_{\beta,c}\right)$. Hence, $d^\ast_{\beta}(C)=d^\ast\left(\varphi_{\beta,e}\right)$ and 
\begin{equation}\label{cotaBCH}
 \Delta(C)=d^\ast(C)=d^\ast\left(\varphi_{\beta,e}\right)=d^\ast\left(\varphi_{\beta,g}\right)\leq d(C),\quad \forall \beta\in \R(C).
\end{equation}
(see \cite[p. 22]{C}). 

\begin{example}\rm{
 Set $q=2$, $n=17$ and take $a_1=0, a_2=1, a_3=3$ as representatives of the $2$-cyclotomic cosets in $\Z_{17}$. Then $A(17)=\{1,3\}$. Let $C$ be the cyclic code with defining set $D_\alpha(C)= C_2(1)=\{1,2,4,8,9,13,15,16\}$ with respect to $\alpha \in U_{17}$, such that $\mathrm{min}(\alpha)=x^8+x^7+x^6+x^4+x^2+x+1$. The  reader may check that $e = x^{16}+x^{15}+x^{13}+x^9+x^8+x^4+x^2+x+1$ is the idempotent generator of $C$, and $M(\varphi_{\alpha,e})=( 1\, 0\, 0\, 1\, 0\, 1\, 1\, 1\, 0\, 0\, 1\, 1\, 1\, 0\, 1\, 0\, 0 )$. Then $d^\ast\left(\varphi_{\alpha,e}\right)=3$. Taking $\beta$ such that $\beta^3=\alpha$, one may check that $d^\ast\left(\varphi_{\beta,e}\right)=4$. Hence $\Delta(C)=d^\ast(C)=d^\ast\left(\varphi_{\beta,e}\right)=4$.
}\end{example}

As an immediate consequence of (\ref{cotaBCH}) we have the following corollary.

\begin{corollary}\label{corolario condicion suficiente dC=Delta C}
Let $C$ be a cyclic code in $\F_q(n)$ and let $e,g\in C$ be the idempotent generator and the generator polynomial of $C$, respectively. For $f\in \{e,g\}$ we have that if $d^\ast(\varphi_{\alpha,f})=\omega(f)$  for some  $\alpha \in U_n$ then $d(C)=\Delta(C)$ and $\alpha\in \R(C)$.
\end{corollary}
\begin{proof}
By hypothesis, $d^\ast(\varphi_{\alpha,f})=\omega(f)\geq d(C)$. Now, if $\beta \in \R(C)$ then $d^\ast(\varphi_{\alpha,f})\leq d^\ast(\varphi_{\beta,f})$ and so (\ref{cotaBCH}) get us the result.
\end{proof}

In the following table we list non trivial cyclic codes of lenght at most 31, satifying the conditions of the corollary above; that is, $d^*\varphi_{\alpha,f}=\omega(f)$ for some $\alpha\in U_n$. Here, $\overline{D(C)}=\Z_n\setminus D(C)$.  Computations were done by using GAP4r7.

\[\begin{array}{|c|c|c|c|}
\hline
 \text{Lenght}& \overline{D(C)}& \dim_\F(C)&d(C) \\ \hline
 7 & C_2(3) & 3& 4 \\ \hline
 & C_2(1) & 3 & 4 \\ \hline
 & C_2(0)\cup C_2(3)& 4& 3 \\ \hline
 & C_2(0)\cup C_2(1)& 4& 3 \\ \hline
 9  & C_2(3) & 2 & 6 \\ \hline 
 & C_2(0) \cup C_2(3) & 3 & 3 \\ \hline 
 & C_2(1) & 6 & 2 \\ \hline 
\end{array}\]

\[\begin{array}{|c|c|c|c|}
\hline
 \text{Lenght}& \overline{D(C)}& \dim_\F(C)&d(C) \\ \hline
 15  & C_2(5) & 2 & 10 \\ \hline 
 & C_2(0) \cup C_2(5) & 3 & 5 \\ \hline 
 & C_2(1) & 4 & 8 \\ \hline 
 & C_2(7) & 4 & 8 \\ \hline 
& C_2(0) \cup C_2(3) & 5 & 3 \\ \hline 
& C_2(0) \cup C_2(7) & 5 & 7 \\ \hline 
 & C_2(0) \cup C_2(1) & 5 & 7 \\ \hline
 & C_2(3) \cup C_2(5) & 6 & 6 \\ \hline 
 & C_2(0) \cup C_2(5) \cup C_2(7) & 7 & 5 \\ \hline 
 & C_2(1) \cup C_2(3) & 8 & 4 \\ \hline 
 & C_2(3) \cup C_2(7) & 8 & 4 \\ \hline 
 & C_2(3) \cup C_2(7) & 8 & 4 \\ \hline 
   & C_2(1) \cup C_2(5) \cup C_2(7) & 10 & 2 \\ \hline 
  21 & C_2(7) & 2 & 14 \\ \hline
  & C_2(3) & 3 & 12 \\ \hline 
  & C_2(9) & 3 & 12 \\ \hline 
  & C_2(0) \cup C_2(7) & 3 & 7 \\ \hline 
  & C_2(0) \cup C_2(3) & 4 & 9 \\ \hline 
  & C_2(0) \cup C_2(9) & 4 & 9 \\ \hline  
  & C_2(3) \cup C_2(7) & 5 & 10 \\ \hline 
  & C_2(7) \cup C_2(9) & 5 & 10 \\ \hline 
  & C_2(0) \cup C_2(3) \cup C_2(9) & 7 & 3 \\ \hline 
  & C_2(1) \cup C_2(7) & 8 & 6 \\ \hline 
  & C_2(5) \cup C_2(7) & 8 & 6 \\ \hline 
  & C_2(1) \cup C_2(9) & 9 & 4 \\ \hline 
  & C_2(3) \cup C_2(5) & 9 & 4 \\ \hline 
  & C_2(0) \cup C_2(5) \cup C_2(9) & 10 & 5 \\ \hline 
   & C_2(0) \cup C_2(1) \cup C_2(3) & 10 & 5 \\ \hline 
  & C_2(5) \cup C_2(7) \cup C_2(9) & 11 & 6 \\ \hline 
   & C_2(0) \cup C_2(1) \cup C_2(7) \cup C_2(9) & 12 & 3 \\ \hline 
%\end{array}\]
%
%\[\begin{array}{|c|c|c|c|}
%\hline
% \text{Lenght}& \overline{D(C)}& \dim_\F(C)&d(C) \\ \hline
 25 & C_2(3)\cup C_2(5) & 5 & 5 \\ \hline
 27 & C_2(9) & 2 & 18 \\ \hline
 & C_2(3) & 5 & 6 \\ \hline
 & C_2(1) & 18 & 2 \\ \hline
 & C_2(0)\cup C_2(9) & 3 & 9 \\ \hline
31 & C_2(1) & 5 & 16 \\ \hline 
& C_2(5) & 5 & 16 \\ \hline 
& C_2(15) & 5 & 16 \\ \hline 
& C_2(0) \cup C_2(1) & 6 & 15 \\ \hline 
& C_2(0) \cup C_2(15) & 6 & 15 \\ \hline 
& C_2(3) \cup C_2(7) & 10 & 6 \\ \hline 
& C_2(5) \cup C_2(11) & 10 & 10 \\ \hline 
 & C_2(1) \cup C_2(3) \cup C_2(15) & 15 & 6 \\ \hline 
 & C_2(1) \cup C_2(5) \cup C_2(11) & 15 & 6 \\ \hline 
 & C_2(1) \cup C_2(7) \cup C_2(15) & 15 & 6 \\ \hline 
 & C_2(5) \cup C_2(9) \cup C_2(15) & 15 & 6 \\ \hline 
 & C_2(0) \cup C_2(1) \cup C_2(3) \cup C_2(7) & 16 & 5 \\ \hline 
 & C_2(0) \cup C_2(1) \cup C_2(11) \cup C_2(15) & 16 & 5 \\ \hline 
 & C_2(0) \cup C_2(1) \cup C_2(5) \cup C_2(15) & 16 & 5 \\ \hline 
 & C_2(0) \cup C_2(3) \cup C_2(5) \cup C_2(11) & 16 & 5 \\ \hline 
 & C_2(0) \cup C_2(5) \cup C_2(7) \cup C_2(11) & 16 & 5 \\ \hline 
 & C_2(0) \cup C_2(3) \cup C_2(5) \cup C_2(11) & 16 & 5 \\ \hline 
\end{array}\]

Let us comment how these results allow us to construct cyclic codes and to compute its apparent distance (or the BCH bound). First, let us observe that for any cyclic code $C$ generated by $e=e^2\in \F_q(n)$ one has that $\varphi_{\alpha, e}$ is an idempotent in $(\Le,\star)$. So,  $\varphi_{\alpha,e}[i]=e(\alpha^i)=0$ if $i\in D_\alpha(C)$ and $1$ otherwise.

Now, let $\{a_1,\dots,a_h\}$ be a complete set of represantives of the $q$-cyclotomic cosets modulo $n$. For each choice $D=\cup_{j=1}^t C_q(a_{i_j})$, with $i_j\in \{1,\dots,h\}$ and $1\leq t\leq h$, we denote by $F_{D}\in \F_q^n$ the vector such that $F_{D}[i]=0$  if $i\in D$ and $1$ otherwise. Then $F_D$ may be viewed as the image under the Fourier transform of the idempotent generator of a cyclic code $C$ in $\F_q(n)$ such that $D=D_\alpha(C)$ with respect to some $\alpha \in U_n$. That is, if $C$ is the cyclic code with defining set $D_\alpha(C)=D$, with respecct to $\alpha\in U_n$, and $e^2=e$ is its idempotent generator then we have that $F_D=\varphi_{\alpha,e}\in \F_q^n$. To compute the apparent distance $d^\ast(C)$ we first consider the set $A(n)=\{a_{i_1},\dots,a_{i_k}\}\subseteq\{a_1,\dots,a_h\}$. Then, for every $j=1,\dots,k$, let $\beta_j\in U_n$ be such that $\beta_j^{a_{i_j}}=\alpha$; recall that this implies $D_{\beta_j}(C)=a_{i_j}\cdot D_\alpha(C)$. The apparent distance of $\varphi_{\beta_j,e}$ is the length of the biggest chain of consecutive zeros (modulo $n$) in $F_{D_{\beta_j}(C)}$ plus $1$. So, $\displaystyle d^\ast(C)=\max_{j=1,\dots,k} d^\ast (F_{D_{\beta_j}(C)})$.

\begin{example}\label{ejemplo calculo de parametros con la matriz F_D(C)}\rm{
 Set $n=21$, $q=2$ and $A(21)=\{1,5\}$. Consider the $2$-cyclotomic cosets: $C_2(0)$, $C_2(1)$, $C_2(3)$, $C_2(5)$, $C_2(7)$, $C_2(9)$ listed in Example~\ref{distancia de bose mayor que dist BCH}. Choose $D=C_2(1)\cup C_2(3)\cup C_2(7)$. Then
\begin{eqnarray}
 F_D&=& \left(1\,0\,0\,0\,0\,1\,0\,0\,0\,1\,1\,0\,0\,1\,0\,1\,0\,1\,1\,1\,1\right). \nonumber
\end{eqnarray}
Let $C=\left\langle e\right\rangle$ be the cyclic code such that $D_\alpha(C)=D$ for some $\alpha\in U_{21}$. Then  $d^\ast(\varphi_{\alpha,e})=5$. We only need to consider $\beta\in U_{21}$ such that $\beta^5=\alpha$. In that case, $D_\beta(C)=5\cdot D_\alpha(C)=C_2(5)\cup C_2(9)\cup C_2(7)$. Then  
\begin{eqnarray}
 F_{D_\beta(C)}&=& \left(1\,1\,1\,1\,1\,0\,1\,0\,1\,0\,0\,1\,1\,0\,0\,0\,1\,0\,0\,0\,0\right). \nonumber
\end{eqnarray}
So $d^\ast(F_{D_\beta(C)})=5$ too.}
Hence $d^\ast(C)=5$ and $\R(C)=\{\beta, \beta^5 \}$.  The reader may check that $C$ has four BCH bounds, $\delta=2,3,4,5$. 
\end{example}

\section{The minimum distance and the BCH bound}

For an arbitrary element $g\in \Le(n)$, which we may view as a polynomial with $\deg(g)\leq n-1$, it is easy to see that the equality $\gcd(g,x^n-1)=\gcd(x^hg,x^n-1)$ holds for any $h\in \{0,\dots,n-1\}$ as $x^h$ and $x^n-1$ are relatively prime polynomials; so, we may write
\begin{equation}\label{mcd}
 m_g=\gcd(x^hg,x^n-1)
\end{equation}
as $m_g$ does not depend on $h$.  For any $h\in \{0,\dots,n-1\}$ we also write
\begin{equation}\label{division}
 x^hg=(x^n-1)f_{g,h}+\overline{x^hg}
\end{equation}
where $0\leq \deg(\overline{x^hg})<n$. Note that if $g\neq 0$ then  $\overline{x^hg}\neq 0$ because $\deg(g)<n$. By using results in \cite{C} and \cite{Ch} (see also \cite[Theorem 8.6.31]{S}) we obtain the following result.

\begin{lemma}\label{propiedades de d^*}
Consider $g\in \Le(n)$ and let $m_g$ be as above. Then
\begin{enumerate}
 \item $d^\ast(g)\leq n-\deg(m_g)$.
\item If $g\mid x^n-1$ then $d^\ast(g)=n-\deg(g)$.
\end{enumerate}
\end{lemma}
\begin{proof}
 \textit{(1)} It comes from the fact that $m_g\mid \overline{x^hg}$ for any $0\leq h\leq n-1$, and from Definition \ref{aparentepol}.
\textit{(2)} By the definition of $d^\ast(g)$ we have that $d^\ast(g)\geq n-\deg(g)$. To get the converse inequality note that $g=s m_g$, for some $s\in\F_q$, and apply \textit{(1)}.
\end{proof}

Now let $C$ be a cylic code in $\F_q(n)$ and let $c\in C$ be any codeword. By (\ref{dist aparente imagen menor que peso}) we have that $d^\ast(\varphi_{\alpha,c})\leq \omega(c)$. We wonder if the equality may occur. Next result will be helpful to find an answer (see \cite[Theorem 4.1]{C} and \cite[Theorem 2]{Ch}).

\begin{lemma}\label{lema de camion y chon}
 Let $C$ be a cyclic code in $\F_q(n)$ and $c\in C$. Then
 $n-\deg\left(m_{\varphi_{\alpha,c}}\right)=\omega(c)$, for all $\alpha\in U_n$.
\end{lemma}

\begin{proof}
We have that $n-\deg\left(m_{\varphi_{\alpha,c}}\right)=|\{\alpha^j\tq \varphi_{\alpha,c}(\alpha^j)\neq 0\}|$. By Remark~\ref{basicos de la DFT} and (\ref{peso y dominio de def}) we are done.
\end{proof}

Note that by Lemma \ref{propiedades de d^*} we have that the apparent distance of any $f\in \Le(n)$ is less than or equal to the number of nonzeros of $m_f$. The following result shows us when the equality holds.

\begin{proposition}\label{igualdad distancia aparente y n-grado}
Consider $f\in \Le(n)$ and let $m_f$ be as in (\ref{mcd}). Then $d^\ast(f)= n-\deg(m_f)$ if and only if there exists $h\in\{0,\dots,n-1\}$ such that $\overline{x^hf}\mid x^n-1$ (equivalently, $\overline{x^hf}$ and $m_f$ are associated polynomials in $\Le[x]$).
\end{proposition}

\begin{proof}
 Suppose first that the equality holds. By definition of apparent distance we know that there exists $h\in\{0,\dots,n-1\}$ such that $d^\ast(f)=n-\deg\left(\overline{x^hf}\right)$. Hence $\deg\left(\overline{x^hf}\right)=\deg\left(m_f\right)$. By (\ref{mcd}) and (\ref{division}) we have that $m_f$ and  $\overline{x^hf}$ have exactly the same set of zeros and hence they are associated polynomials, or equivalently, $\overline{x^hf}\mid x^n-1$.

Conversely, suppose that there exists $h\in\{0,\dots,n-1\}$ such that $\overline{x^hf}\mid x^n-1$. Again by (\ref{division}) and (\ref{mcd}), $\overline{x^hf}$ and $m_f$ must be associated polynomials. By definition of apparent distance we have that $d^\ast(f)=d^\ast\left(\overline{x^hf}\right)$ and by Lemma \ref{propiedades de d^*}\textit{(2)}, $d^\ast\left(\overline{x^hf}\right)=n-\deg\left(\overline{x^hf}\right)$. The result follows immediately.
\end{proof}

Now we deal with our first problem. We are going to present some results that give theoretical characterizations for a given cyclic code to satisfy the equality $d(C)=\Delta(C)$.

%The next result is a theoretical characterization of cyclic codes such that its minimum distance and its apparent distance coincide. We will use it to give practical sufficient conditions.

\begin{theorem}\label{teo principal}
Let $n$ be a positive integer, $p$ a prime number and $q$ a power of $p$. Assume that $\gcd(n,q)=1$. Consider the field $\F_q$ and an extension $\Le|\F_q$ such that $U_n \subseteq \Le$. Let $C$ be a cyclic code in $\F_q(n)$. Then $d(C)=\Delta(C)$ if and only if there exists a polynomial $f \in \Le(n)$, such that
\begin{enumerate}
 \item $d^\ast(f)=d^\ast(C)$.
\item $d^\ast (f)=n-\deg(m_{f})$.
\item $\varphi^{-1}_{\alpha,f}\in C$, for some $\alpha\in \R(C)$.
\end{enumerate}
Moreover, in this case, there exists $h\in\{0,\dots,n-1\}$ such that $\overline{x^h f}\mid x^n-1$.
\end{theorem}

\begin{proof}
 First, suppose that $d(C)=\Delta(C)$. Then we have that $d(C)=d^\ast(C)$. Let $c\in C$ such that $\omega(c)=d(C)$, consider $\alpha \in \R(C)$ and set, as in (\ref{mcd}), $m_{\varphi_{\alpha,c}}=\gcd(\varphi_{\alpha,c},x^n-1)$. By definition of apparent distance and by applying results above, we have that
	\[\omega(c)\geq d^\ast(\varphi_{\alpha,c})\geq d^\ast_\alpha(C)=d^\ast(C)=d(C)=\omega(c)=n-\deg\left(m_{\varphi_{\alpha,c}}\right).\]
Hence $d^\ast(\varphi_{\alpha,c})=d^\ast(C)$, since $d^\ast\left(\varphi_{\alpha,c}\right)=n-\deg\left(m_{\varphi_{\alpha,c}}\right)$. So, $f=\varphi_{\alpha,c}$ satisfies all required conditions.

Conversely, suppose there exists $f\in \Le(n)$ satisfying conditions \textit{(1 -- 3)} of the statement. By Lemma~\ref{lema de camion y chon} and the definition of minimum distance, we have that $d^\ast(f)=\omega(\varphi^{-1}_{\alpha,f})\geq d(C)$. Then by Condition \textit{(1)}, $d^\ast(C)\geq d(C)$, and hence by (\ref{cotaBCH}), $\Delta(C)=d(C)$.

 The final assertion follows directly from Proposition \ref{igualdad distancia aparente y n-grado}.
\end{proof}

So, to check if a code satisfies the conditions in the theorem above, Proposition \ref{igualdad distancia aparente y n-grado} shows us that we have to focus on properties of some divisors of $x^n-1$. After Corollary \ref{resultado principal con el D(C)} we will make some comments about complexity in order to consider those divisors.

% \begin{corollary}\label{resultado principal con divisores de x^n-1} Hypotheses as in the theorem above. Consider an intermediate field $\F_q\subseteq \K\subseteq \Le$, and let $C$ be a cyclic code in $\F_q(n)$. Suppose that there exists $k\in\{0,\dots,n-1\}$ and a divisor $g\mid x^n-1$, in $\K[x]$, such that setting $f=\overline{x^kg}$, the following conditions hold.
% \begin{enumerate}
%  \item $d^\ast(f)=d^\ast(C)$.
% \item $\varphi^{-1}_{\alpha,f}\in C$.
% \end{enumerate}
% Then $d(C)=\Delta(C)$.
% 
% Sufficiency is obtained when $\K=\Le$.
% \end{corollary}

\begin{corollary}\label{resultado principal con divisores de x^n-1} Let $C$ be a cyclic code in $\F_q(n)$. Then $d(C)=\Delta(C)$ if and only if there exist $k\in\{0,\dots,n-1\}$ and a divisor $g\mid x^n-1$, in $\Le[x]$, such that setting $f=\overline{x^kg}$, the following conditions hold
\begin{enumerate}
 \item $d^\ast(f)=d^\ast(C)$ .
\item $\varphi^{-1}_{\alpha,f}\in C$, for some $\alpha\in \R(C)$.
\end{enumerate}
\end{corollary}

\begin{proof}
 Set $h=n-k$. Then $g=\overline{x^hf}$ and the result follows from Proposition~\ref{igualdad distancia aparente y n-grado} and the theorem above.
\end{proof}

We note that, in the setting of the previous corollary, it may happen that there exist $\alpha,\beta\in U_n$ such that $\varphi^{-1}_{\alpha,f}\in C$ but $\varphi^{-1}_{\beta,f}\notin C$. 

%By Remark~\ref{basicos de la DFT} 
We may rewrite the condition \textit{(3)} in Theorem~\ref{teo principal} or \textit{(2)} in Corollary~\ref{resultado principal con divisores de x^n-1}, as follows.

\begin{corollary}\label{resultado principal con el D(C)}
 Let $C$ be a cyclic code in $\F_q(n)$. Then $d(C)=\Delta(C)$ if and only if there exist $k\in\{0,\dots,n-1\}$ and a divisor $g\mid x^n-1$, in $\Le[x]$, such that  the following conditions hold.
\begin{enumerate}
 \item $d^\ast(g)=d^\ast(C)$, and setting $f=\overline{x^kg}$,
\item\label{condicion supp} $supp(f)\subseteq \Z_n\setminus D_\alpha(C)$, for some $\alpha\in \R(C)$,
%\item $supp(f)=supp(f^{q^j})$ (equivalently, $supp(f)$ is union of $q$-cyclotomic cosets).
\item\label{condicion F} $(f(\alpha^j))^q=f(\alpha^j)$, for any $j\in \{0,\dots,n-1\}$.
\end{enumerate}
\end{corollary}

\begin{proof}
From Remark~\ref{basicos de la DFT}, it comes immediately that condition \textit{(2)} in Corollary~\ref{resultado principal con divisores de x^n-1} holds if and only if conditions \textit{(2)}+\textit{(3)} of this corollary hold. 
\end{proof}

Given a linear code $C$ of length $n$, we wonder about how difficult is to check the equality $\Delta(C)=d(C)$; in other words, using our previous results, how difficult is to find a polynomial satisfying the required conditions?

%We are going to comment how one may apply the corollaries above. Let $C$ be a cyclic code of length $n$, having $\Delta(C)=\delta$. We want to know if $d(C)=\Delta(C)$.

To apply any of the corollaries above we have to compute the divisors $g \mid x^n-1$ in  $\Le[x]$ with $\deg(g)=n-\Delta(C)$. This means that we have to check at most $h \cdot \binom{n}{n-\Delta(C)}$ polynomials, where $h=|A(n)|$. Clearly, if $\Delta(C)$ is not a ``big'' number we may check all divisors in $\Le[x]$. In case that $\Delta(C)$ was a ``big'' number, we could reduce it by taking an intermediate field, $\F_q\subset \K \subset \Le$, where the number of divisors of $x^n-1$ (in $\K[x]$) is smaller. However, in that case, our searching of codes would not be exhaustive.

For example, consider the binary cyclic code $C$ of length $45$ with $D_\alpha(C)=C_2(3)\cup C_2(5)$, for some $\alpha \in U_{45}$. One may see that $\Delta(C)=3$ and $\dim(C)=35$. Consider $A(45)=\{1,7\}$. To check any of our corollaries above we have to consider $2\binom{45}{42}$-polynomials (note that $2^{14}<2\binom{45}{42}<2^{15}$) so our method works. On the other hand, for codes with apparent distance greater than $5$, we might choose to consider the factors of $x^{45}-1$ in an intermediate ring. For example, in $\F_{2^4}[x]$ there are 15 factors of degree 1 and  10 factors of degree 3. No more than 50 computations. Essentially the same happens in $\F_{2^6}[x]$.

%{\color{green}
%\begin{remark}\label{la invTFD en F para una raiz lo es para todas}\rm{
% The case $\K=\F_q=\F_2$ of our comments above is especially simple as we only have to check that $supp(\overline{x^kg})$ is a union of $2$-cyclotomic cosets; so, condition \textit{(3)} in the corollary above need not be checked.
%}\end{remark}}

Now we give another sufficient condition to characterize cyclic codes whose apparent distance reaches its minimum distance.

\begin{corollary}
 Let $C$ be a cyclic code in $\F_q(n)$ with generator idempotent $e\in C$. If there exist $h\in\{0,\dots,n-1\}$ and $\alpha \in U_n$  such that $\overline{x^h\varphi_{\alpha,e}}\mid x^n-1$ then $d(C)=\Delta(C)$ and $\alpha \in \R(C)$.
\end{corollary}

\begin{proof}
 From Proposition \ref{igualdad distancia aparente y n-grado} and Lemma \ref{lema de camion y chon} we may deduce that $d^\ast\left(\varphi_{\alpha,e}\right)=n-\deg\left(m_{\varphi_{\alpha,e}}\right) =\omega(e)$. So, the result follows directly from Corollary~\ref{corolario condicion suficiente dC=Delta C}.
\end{proof}

%\begin{remark}\label{imagen inversa de la TFD en F(n) para toda raiz}\rm{
%It is easy to check the following fact that we will use repeatedly.  For any $g\in \Le(n)$ and $\alpha,\beta\in U_n$, it happens that $\varphi^{-1}_{\alpha,g}\in \F_q(n)$ if and only if $\varphi^{-1}_{\beta,g}\in \F_q(n)$.
%}\end{remark}

The previous results give us conditions for a cyclic code $C$ to satisfy the equality $d(C)=\Delta(C)$. Now we deal with our second problem, that is, the construction of such kind of codes.

%From Corollary~\ref{resultado principal con divisores de x^n-1}, we have that in order to construct a cyclic code $C$ satisfying that $\Delta(C)=d(C)$, we may focus on the inverse Fourier transform of the divisors of $x^n-1$ in some intermediate ring $\K[x]$.

%{\color{green} The following corollary shows us how to do it. As we have noted above, we are considering to impose restrictions on the minimum distance or the intermediate field in order to reduce the number of computations.}

\begin{corollary}\label{construccion codigos delta=distancia aparente}
Consider an intermediate field $\F_q\subseteq \K\subseteq \Le$, let $g\in \K[x]$ be a divisor of $x^n-1$ and $\beta\in U_n$. If $\varphi^{-1}_{\beta,_{\overline{x^kg}}}$ belongs to $ \F_q(n)$, for some $k\in \{0,\dots,n-1\}$, then the family of permutation equivalent cyclic codes $\left\{C_\alpha=\left(\varphi^{-1}_{\alpha,_{\overline{x^kg}}}\right) \tq \alpha \in U_n\right\}$ satisfies $\Delta(C_\alpha)=d(C_\alpha)$ for all $\alpha\in U_n$.  Moreover, in this case, $\dim_{\F_q}(C_\alpha)=|supp(g)|$, for all $\alpha \in U_n$.
\end{corollary}

\begin{proof}
Fix $\alpha\in U_n$. Set $f=\overline{x^kg}$ and let $e\in \F_q(n)$ be the idempotent generator of the ideal $C=\left(\varphi^{-1}_{\alpha,f}\right)$ in $\F_q(n)$ (see Remark~\ref{basicos de la DFT}(4)). It is easy to check that $supp(\varphi_{\alpha,e})=supp(f)=\Z_n\setminus D_\alpha(C)$ and hence, $d^\ast(\varphi_{\alpha,e})=d^\ast(f)$. On the one hand, by Proposition~\ref{igualdad distancia aparente y n-grado} and Lemma~\ref{lema de camion y chon} one has that $d^\ast f=n- deg(m_f)= \omega\left(\varphi^{-1}_{\beta,f}\right)\geq d(C)$. On the other hand, by (\ref{cotaBCH}), $d^\ast f=d^\ast\left(\varphi_{\alpha,e}\right)\leq d^\ast_\alpha(C)\leq d^\ast(C)\leq d(C)$. So we are done.
\end{proof}

Then, in order to construct codes with the desired property we need to find a divisor $g$ of $x^n-1$ satisfying the condition \textit{(2)} in Corollary \ref{resultado principal con divisores de x^n-1}. However, in the case $\K=\F_2$, it is clear that $g\in (\F_2^n,\star)$ is always an idempotent, and so, we only have to check that $supp(g)$ is union of $2$-cyclotomic cosets (see Lemma~\ref{TFDenFq}).

Let us show by an example how the combination of Corollary~\ref{resultado principal con el D(C)} and Corollay~\ref{construccion codigos delta=distancia aparente} works.

\begin{example}\label{diseñando ciclicos}\rm{
 Set $q=2$, $n=45$.  In this case $A(45)=\{1,7\}.$ Take $g=x^{40}+x^{39}+x^{38}+ x^{36}+x^{35}+x^{32}+x^{30}+x^{25}+x^{24}+x^{23}+x^{21}+ x^{20}+x^{17}+ x^{15}+x^{10}+x^9+ x^8+x^6+x^5+x^2+1$. One may check that $g\mid x^{45}-1$ in $\F_2[x]$ (so that $\K=\F_2$). To find the parameter $k$ mentioned in the corollary above, we may analize the vector $M(g)$ or we may fix $\beta \in U_{45}$ (as instance, such that $\mathrm{min}(\beta)=x^{12}+x^3+1$) and compute $g(1)$ and $g(\beta^3)$, because $D_\beta(g)=\Z_{45}\setminus \left( C_2(0)\cup C_2(3)\right)$. Let us choose the last alternative. Since $g(1)=1$ and $g(\beta^3)=\beta^{30}$ then $k=5$ will work because setting $f=x^{5}g$ we have that $f(1)=1$, $f\left(\beta^3\right)=(\beta^3)^5\beta^{30}=\beta^{45}=1$ and then $f\left(\beta^6\right)=f\left(\beta^{12}\right)=f\left(\beta^{24}\right)=1$, as $C_2(3)=\{3,6,12,24\}$. So that,  $\varphi^{-1}_{\alpha,f}\in \F(45)$, for all $\alpha\in U_{45}$. Now set $C=(\varphi^{-1}_{\beta,f})$. Then $D_\beta(C)=C_2(1)\cup C_2(3)\cup  C_2(9)\cup C_2(21)=\Z_{45}\setminus supp(M(f))$ and, by analizing $M(g)$ or $M(f)=F_{D_\beta(C)}$ as in Example~\ref{ejemplo calculo de parametros con la matriz F_D(C)}, we have that $5=d(C)=\Delta(C)$  and $\dim(C)=21$.

As $supp(x^5g)=\Z_{45}\setminus D_\beta(C)$, one may see that there are three subsets that determines $d^\ast(C)$; to wit, $\{1,2,3,4\}$, $\{16,17,18,19\}$ and $\{31, 32, 33, 34\}$. We choose $\{1,2,3,4\}\subset D_\beta(C)$ and construct the code $C'$ such that $D_\beta(C')=D_\beta(C)\setminus C_2(21)$. Note that $C$ is a subcode of $C'$, because $D_\beta(C')\subset D_\beta(C)$. Now one has that $C'$ satisfies the conditions in Corollary~\ref{resultado principal con divisores de x^n-1}, because $d^\ast (C)=5=d^\ast(f)$ and $\varphi^{-1}_{\alpha, f}\in C \subset C'$, so that $5=d(C')=\Delta(C')$  and $\dim(C')=25$, that is, $C'$ has better parameters than $C$.\medskip
}\end{example}

In the next section (see, as instance, Example~\ref{ejemplo extendido n=45}) we will refine this type of construction to obtain BCH codes $C$ such that $\Delta(C)=d(C)$.  Now we continue with the construction of codes $C$ satisfying that $\Delta(C)=d(C)$.

\begin{corollary}\label{NO bch cond suficiente true}
Consider an intermediate field $\F_q\subseteq \F_{q'}\subseteq \Le$, let $h$ be an irreducible factor of $x^n-1$ in $\F_{q'}[x]$ with defining set $D_\alpha(h)$ for some $\alpha\in U_n$. Set $g=(x^n-1)/h$.  If there are positive integers $j,t$ such that $g(\alpha^j)=\alpha^t$ and $\gcd \left(j,\frac{n}{\gcd(q-1,n)}\right)\mid t$ then there exists a $q$-ary code of length $n$ whose BCH bound equals its minimum distance.
\end{corollary}

\begin{proof}
 By hypothesis, the congruence (in $X$),  $$\frac{q-1}{\gcd(q-1,n)}jX\equiv -\frac{q-1}{\gcd(q-1,n)}t \mod \frac{n}{\gcd(q-1,n)}$$ has a solution $X=k$, with $0\leq k\leq \frac{n}{\gcd(q-1,n)}$. Then $(q-1)(jk+t)\equiv 0 \mod n$, which means that $q(jk+t)\equiv jt+k \mod n$, and hence $\overline{x^kg}(\alpha^j)=\alpha^{jk+t}\in \F_q$. Clearly, for any $jq'^{a}\in D_{\alpha}(h)$ we have $jq'^{a}k+tq'^{a}\equiv q'^{a}( jk+t)\equiv jk+t \mod n$, so that $\overline{x^kg}(\alpha^{jq'^{a}})\in \F_q$. As $\overline{x^kg}(\alpha^i)=0$ for all $i\in \Z_n\setminus D_\alpha(h)$, we may apply Corollary~\ref{construccion codigos delta=distancia aparente} to get the desired result. More precisely, the code $C=\left(\varphi^{-1}_{\alpha,\overline{x^k g}}\right)\subseteq \F_q(n)$ satisfies the required conditions.
\end{proof}

\begin{corollary}\label{NO bch cond suficiente true con n=q^m-1}
  Let $n=2^m-1$, for some $m\in \N$. There exist at least $\frac{\phi\left(n\right)}{m}$ binary codes of length $n$ whose BCH bound equals its minimum distance.
\end{corollary}
\begin{proof}
 We are going to apply the corollary above with $2=q=q'$.  Take $\Le=\F_{2^m}$. For each $0<j<n$, coprime with $n$, we consider the $2$-cyclotomic coset $C_2(j)$, which has exactly $m$ elements. Consider $\alpha \in U_n$.  Let $h|x^n-1$ be the polynomial in $\F_q[x]$, such that $D_\alpha(h)=C_2(j)$ and $g_j=(x^n-1)/h$. By hypothesis, $\alpha$ is a primitive element for $\Le$, so that $g_j(\alpha^j)=\alpha^k$ for some $k\in \Z_n$. The condition $\gcd\left(j,n\right)\mid k$ holds obviously. So that there exists a binary code of length $n$ whose BCH bound equals its minimum distance. Moreover, by Corollary~\ref{construccion codigos delta=distancia aparente} the family of codes $\{C_j=\left(\varphi^{-1}_{\alpha,\overline{x^k g_j}}\right)\mid \gcd(j,n)=1\}$ satisfies that $d(C_j)=\Delta(C_j)$ for any $j$. To compute the number of different codes in that family we consider the set $B=\{C_2(j) \tq j \in \Z_n, \gcd(j,n)=1\}$. One may check that $|B|=\frac{\phi\left(n\right)}{m}$.  Let $C_q(j)\neq C_q(j')\in B$. If $\alpha \in U_n$ and $h, h'$ are the divisors of $x^n-1$ with $D_\alpha(h)=C_q(j)$ and $D_\alpha(h')=C_q(j')$  then $g_j=(x^n-1)/h$ and $g_{j'}=(x^n-1)/h'$ have the same degree, and hence $supp(g_j)\neq supp(g_{j'})$ because they are binary polynomials. Since $D(C_j)=supp(\overline{x^k g_j})$ the result comes immediately.
\end{proof}

\begin{example}\label{tabla n=15}\rm{
 Set $q=2$, $n=15$. Then $A(15)=\{1,7\}$. By Corollary~\ref{NO bch cond suficiente true con n=q^m-1} there exist at least two codes such that its BCH bound equals its minimum distance (they will be determined by the polynomials $g_3$ and $g_4$ defined below). Denote the irreducible factors of $x^{15}-1$ in $\F_2[x]$, by $h_1=\Phi_2$, $h_2=\Phi_3$, $h_3=x^4+x+1$, $h_4=x^4+x^3+1$ and $h_5=\Phi_5$, where $\Phi_j$ denotes the $j$-th cyclotomic polynomial. Setting $g_i=\frac{x^n-1}{h_i}$, $i=1, \dots, 5$, we apply the corollaries above (with $\K=\F_2$) as follows.
 
Consider the factor $g_2$. Then one may check that in this case $\varphi^{-1}_{\alpha,{\overline{x g_2}}}=x^{10}+x^5 \in \F_2(15)$, for all $\alpha\in U_{15}$. The cyclic code $C$ generated by $x^{10}+x^5$ satisfies $\dim(C)=10$ and $\Delta(C)=2=d(C)$. Now let us fix $\alpha \in U_{15}$ such that $h_3=\mathrm{min}(\alpha)$ and $h_4=\mathrm{min}(\alpha^{13})$, where $\mathrm{min}(\alpha^{t})$ denotes the minimal polynomial of $\alpha^t$ in $\F_2[x]$. Then  $\varphi^{-1}_{\alpha,{\overline{x g_3}}}=x^{14}+x^{13}+x^{11}+x^7 \in \F_2(15)$ and  $\varphi^{-1}_{\alpha,{\overline{x^3 g_4}}}=x^8+x^4+x^2+x \in \F_2(15)$. This gives us the table 

\[\begin{array}{|l|l|l|} \hline 
 \text{Generator}&  \text{Dimension} & \Delta=d \\  \hline 
 \varphi^{-1}_{\alpha,{\overline{x g_2}}}   & 10 & 2 \\  \hline 
 \varphi^{-1}_{\alpha,{\overline{x g_3}}}    & 8 & 4 \\  \hline 
 \varphi^{-1}_{\alpha,{\overline{x^3 g_4}}} & 8 & 4 \\  \hline 
\end{array}\]

The polynomial $g_1$ gets an improper code. In the case of $g_5$, as $D_\alpha(g_5)=\Z_{15}\setminus C_2(3)$, it happens that, $g_5\left(\alpha^3\right)=\alpha^{14}$, so the conditions of Corollary~\ref{NO bch cond suficiente true} are not satisfied. 

After inspecting the divisors of $x^{15}-1$ in $\F_2[x]$ we find more interesting codes. For instance, one may check that the polynomial $h_2h_3h_5$ satisfies the conditions of Corollary~\ref{construccion codigos delta=distancia aparente}, with $k=0$, and hence it yields a code, say $C'$, such that $\Delta(C')=d(C')=5$ and $\dim (C')=7$. 
}\end{example}

\begin{example}\label{tabla n=21}\rm{
 Set $q=2$ and $n=21$. Denote the irreducible factors of $x^{21}-1$ in $\F_2[x]$ by $h_1=\Phi_2$, $h_2=\Phi_3$, $h_3=x^3+x+1$, $h_4=x^3+x^2+1$,  $h_5=x^6+x^4+x^2+x+1$ and $h_6=x^6+x^5+x^4+x^2+1$.

Set $g_i=\frac{x^n-1}{h_i}$, $i=1,\dots, 6$,  and fix $\alpha\in U_{21}$ such that $\mathrm{min}(\alpha)=h_6$.  We apply Corollary \ref{NO bch cond suficiente true} as above (with $\K=\F_2$) to get the following table of binary codes of length 21 whose BCH bound equals its minimum distance. We complete with another one satisfying the conditions of Corollary~\ref{construccion codigos delta=distancia aparente}.

\[\begin{array}{|l|l|l|} \hline 
 \text{Generator}&  \text{Dimension} & \Delta=d \\  \hline 
 \varphi_{\alpha,\overline{xg_2}}^{-1}&14& 2 \\ \hline
\varphi_{\alpha,g_3}^{-1}&12&3 \\ \hline
\varphi_{\alpha,\overline{x^3g_4}}^{-1}&12&3 \\ \hline
\varphi_{\alpha,\overline{xg_5}}^{-1}&8&6 \\ \hline
\varphi_{\alpha,\overline{x^5g_6}}^{-1}&8&6 \\ \hline
\varphi_{\alpha,\overline{h_1h_3h_5h_6}}^{-1}&10&5 \\ \hline
\end{array}\]

% In fact  $\varphi^{-1}_{\alpha,_{\overline{h_1h_3h_5h_6}}}=x^{14}+x^{12}+x^7+x^6+x^3
%  \in \F_2(15)$.
}\end{example}

\section{Applications: Constructing BCH codes whose minimum distance equals their apparent distance}

The following result allows us to construct BCH codes $B_q(\alpha,\delta,b)$ for which $d(B_q(\alpha,\delta,b))=\Delta(B_q(\alpha,\delta,b))=\delta$. We recall that the ideal generated by a polynomial $g\in \F_q(n)$ is denoted by $(g)$.

\begin{theorem}\label{base true distance}
Let $n$ be a positive integer, $p$ a prime number, $q$ a power of $p$ and $U_n$ the set of  primitive $n$-th roots of unity. Assume that $\gcd(n,q)=1$. Consider the fields $\F_q\subseteq \K\subseteq \Le$ such that $U_n\subset \Le$.  Let $g\in \K[x]$ be a divisor of $x^n-1$. If there exist $k\in \{0,\dots,n-1\}$ and $\beta\in U_n$ such that $\varphi^{-1}_{\beta,{\overline{x^kg}}}\in \F_q(n)$ then there exists a family of permutation equivalent BCH codes $\{C_\alpha=B_q(\alpha, \delta, b) \tq \alpha \in U_n\}$ with $\delta=n-\deg(g)$ and $b\in \Z_n$,  such that $\delta=\Delta(C_\alpha)=d(C_\alpha)$ and $\varphi^{-1}_{\alpha,{\overline{x^kg}}} \in C_\alpha$.
\end{theorem}

\begin{proof}
Set $g=\sum_{i=0}^{n-1}a_ix^i$ and suppose that there there exist $k\in \{0,\dots,n-1\}$ and $\beta\in U_n$ such that $\varphi^{-1}_{\beta,{\overline{x^kg}}}\in \F_q(n)$. Let $f=\overline{x^k g}$ and consider $\alpha \in U_n$. By Lemma \ref{propiedades de d^*}, $d^\ast(g)=n-\deg(g)$. Clearly $m_f=g$ and $d^\ast(f)=d^\ast(g)$. We collect $T=\bigcup_{j=\deg(g)+k+1}^{n+k-1}C_q(\overline{j})$, where $\overline{j}$ is the canonical representative of $j$ module $n$, and $\varepsilon =\sum_{i=0}^{n-1}r_ix^i$ such that $r_i=0$ if $i\in T$ and $1$, otherwise.

We claim that $d^\ast(g)= d^\ast(\varepsilon)$. From the definition of $\epsilon$ one has that $d^\ast(g) \leq d^\ast(\varepsilon)$. We are going to see the reverse inequality. By Remark~\ref{basicos de la DFT}(3), as $\varphi^{-1}_{\beta,{f}}\in \F_q(n)$ we have that $supp\left(f\right)$ is union of $q$-cyclotomic cosets modulo $n$. So, for any $j\in \{\deg(g)+k+1,\dots, n+k-1\}$ we have that $C_q(\overline{j})\cap supp\left(f\right)=\emptyset$ and hence $T\subseteq \Z_n\setminus  supp\left(f\right)$, which means that $ supp\left(f\right)\subseteq supp\left(\varepsilon\right)$ and hence $d^\ast(g)=d^\ast(f)\geq d^\ast(\varepsilon)$.

By construction, $supp(\varepsilon)$ is union of $q$-cyclotomic cosets, so $\varphi^{-1}_{\alpha,\varepsilon}\in \F_q(n)$ (see Lemma~\ref{TFDenFq}). We set $C=\left(\varphi^{-1}_{\alpha,\varepsilon}\right)$ in $\F_q(n)$. We are going to see that $C$ satisfies the conditions \textit{(1)} and \textit{(2)} of Corollary~\ref{resultado principal con divisores de x^n-1}.

 \textit{(1)} We have already seen that $d^\ast(f)=d^\ast(g)=d^\ast(\varepsilon)$. Now, by Proposition~\ref{igualdad distancia aparente y n-grado} and Lemma~\ref{lema de camion y chon} one has that $d^\ast f=n- deg(m_f)= \omega\left(\varphi^{-1}_{\alpha,f}\right)\geq d(C)$. On the other hand, by (\ref{cotaBCH}), $d^\ast f=d^\ast(\epsilon)=d^\ast\left(\varphi_{\alpha,\varphi^{-1}_{\alpha,\epsilon}}\right)\leq d^\ast_\alpha(C)\leq d^\ast(C)\leq d(C)$.Therefore $d^\ast(C)=d^\ast(f)$. 
 
 \textit{(2)} Since $supp(f)\subseteq supp(\varepsilon)$, we have that $f\star \varepsilon=f$, and then $\varphi^{-1}_{\alpha,f}\cdot \varphi^{-1}_{\alpha,\varepsilon}=\varphi^{-1}_{\alpha,f}$, which means that  $\varphi^{-1}_{\alpha,f}\in (\varphi^{-1}_{\alpha,\varepsilon})=C$ (see also Remark~\ref{basicos de la DFT}(4)). So that, conditions of Corollary~\ref{resultado principal con divisores de x^n-1} are satisfied, and hence $d(C)=\Delta(C)$.

Finally, to see that $C$ is a BCH code with designed distance $\delta=\Delta(C)$, we note that, any $q$-cyclotomic coset $Q\subseteq supp(\varepsilon)=D_\alpha(C)$ verifies that  $Q\cap \{deg(g)+k+1,\dots,n+k-1\}\neq\emptyset$. So, as we mentioned in Section~\ref{preliminares}, this means that $C$ is a BCH code with $b=deg(g)+k+1$ and designed distance $\delta=\Delta(C)=n-deg(g)$.
\end{proof}

 The theorem above gives us a method to transform a given cyclic code $C=(g)$, with $d(C)=\Delta(C)$ into another code with higher dimension; in fact, we can get a new BCH code. The key idea is to consider as generator $\varepsilon$ instead of $g$ via the definition of $T$. This definition may be done in different ways that can drives us to different BCH codes. All these ideas are shown in the next example.

\begin{example}\label{ejemplo extendido n=45}\rm{
 We continue with the code $C$ showed in Example~\ref{diseñando ciclicos}. Recall that  $q=2$, $n=45$ and $C$ is the cyclic code with $D_\beta(C)=C_2(1)\cup C_2(3)\cup  C_2(9)\cup C_2(21)$, where $\beta \in U_{45}$ is such that $\mathrm{min}(\beta)=x^{12}+x^3+1$. 
 %Suppose $C'$ is generated by an idempotent $e'$, with $g'=\varphi_{\beta,e'}$ and $f'=x^5g'$. 
 Following the proof of the previous theorem we have that $T=C_2(1)\cup C_2(3)$ and set $\varepsilon=\sum_{i\not\in T}x^i$. Then $C''=\left(\varphi^{-1}_{\beta,\varepsilon}\right)$ has $D_\beta(C'')=C_2(1)\cup C_2(3)$; so that it is the BCH code  $B_2(\beta,5,1)$ of dimension  $29$ such that $d(C'')=\Delta(C'')=5$. This code has even better parameters than $C'$ (see Example~\ref{diseñando ciclicos}).

It is also possible to obtain, from the code $C'$, the BCH code $B_2(\beta, 5,16)$ with $d(B_2(\beta, 5,16))=\Delta(B_2(\beta, 5,16))=5$ and dimension $29$, by taking $T'=C_2(1)\cup C_2(9)$.
}\end{example}

The following theorem is a classical result on the theory of BCH codes.
\begin{theorem}[\cite{S}]\label{teorema clasico}
 Let $h,m\in\N$. A BCH code $C$ of length $n=q^m-1$ and designed distance $\delta=q^h-1$ over $\F_q$ satisfies that $d(C)=\Delta(C)$. 
\end{theorem} 

Now let us show some examples of construction of new BCH codes.

\begin{example}\rm{
 Set $q=2$ and $n=15$. Consider the polynomial $g=g_3$ in Example~\ref{tabla n=15}; that is $g=x^{11}+x^{8}+x^{7}+x^{5}+x^{3}+x^{2}+x+1$. Then, its coefficient vector is
\[ M(g)=(1,1,1,1,0,1,0,1,1,0,0,1,0,0,0)\]
and we may check that $d^\ast(g)=4$. We know that $\varphi^{-1}_{\alpha,g} \not\in \F_2(n)$ for all $\alpha \in U_{15}$, because $C_2(7)$ is not contained in $supp(g)$ (see Lemma~\ref{TFDenFq}). However, the polynomial $\overline{xg}$ with coefficient vector
\[ M(\overline{xg})=(0,1,1,1,1,0,1,0,1,1,0,0,1,0,0)\]
satisfies that $\varphi^{-1}_{\alpha,_{\overline{xg}}} \in \F_2(n)$ for all $\alpha \in U_{15}$. Let us fix $\alpha\in U_{15}$. Then $C=(\varphi^{-1}_{\alpha,_{\overline{xg}}})$ is a binary code with $d(C)=d^\ast(C)=4$ and $\dim_{\F_2}(C)=8$ (see Corollary~\ref{construccion codigos delta=distancia aparente}). But, $C$ is not a BCH code. Following the ideas in Theorem~\ref{base true distance} we may replace $0$'s by $1$'s in the suitable places to get the vector
\[ M(\varepsilon)=(0,1,1,1,1,1,1,0,1,1,1,0,1,0,0)\]
such that $C'=\left(\varphi^{-1}_{\alpha,\varepsilon}\right)$ is a BCH code in $\F_2(n)$, with $d(C')=d^\ast(C')=\delta=4$ and $\dim_{\F_2}(C')=10$. Clearly, this code cannot be considered in Theorem~\ref{teorema clasico}.
}\end{example}

We finish by extending Corollary~\ref{NO bch cond suficiente true} to BCH codes.

\begin{corollary}\label{bch cond suficiente true}
Consider an intermediate field $\F_q\subseteq \F_{q'}\subseteq \Le$, let $h$ be an irreducible factor of $x^n-1$ in $\F_{q'}[x]$ with defining set $D_\alpha(h)$ for some $\alpha\in U_n$ and $g=(x^n-1)/h$.  If there are positive integers $j,t$ such that $g(\alpha^j)=\alpha^t$ and $\gcd \left(j,\frac{n}{\gcd(q-1,n)}\right)\mid t$ then there exists a BCH code of designed distance $\delta$, $C=B_q(\alpha,\delta,b)$, such that $\delta=\Delta(C)=d(C)=\deg(h)$, for certain $b\in\Z_n$.
\end{corollary}

\begin{proof}
Comes immediately from Corollary~\ref{NO bch cond suficiente true} together with  Theorem~\ref{base true distance}.
\end{proof}

\begin{example}\rm{ We continue with the codes  determined by the polynomials $g_2$, $g_3$ and $g_4$ in Example~\ref{tabla n=15}. Recall that in this case $\alpha \in U_{15}$ satisfies that $\mathrm{min}(\alpha)=h_3$. By applying the ideas contained in the proof of Theorem~\ref{base true distance}, one may obtain the following BCH codes whose minimum distance equals the maximum of their BCH bounds.

It is possible to modify the defining set, w.r.t. $\alpha$, of a cyclic code in order to obtain a defining set for a new code with higher dimension. In this case we will say that the original one was dimensional-extended to the new one. For example, $\left(\varphi^{-1}_{\alpha,{\overline{xg_2}}}\right)$ in $\F_q(15)$ has dimension 10 and it can be dimensional-extended to the codes $B_2(\alpha,2,0)$ of dimension $14$ and $B_2(\alpha,2,3t)$ of dimension $11$, for $t=1,2,3$. The cyclic code determined by $g_3$, that is $\left(\varphi^{-1}_{\alpha,{\overline{xg_3}}}\right)$ in $\F_q(15)$, has dimension 8 and it may be dimensional-extended to $B_2(\alpha,4,13)$ of dimension $10$. Finally, from $\left(\varphi^{-1}_{\alpha,\overline{x^3g_4}}\right)$, with dimension 8, we get $B_2(\alpha,4,0)$ of dimension $10$.

Note that the dimensional-extended BCH codes associated to $g_2,g_3$ and $g_4$ are not considered in the classical result \ref{teorema clasico}.
 There is another interesting code which has not been considered: the code $\left(\varphi_{\alpha,\overline{h_1 h_2h_3h_5}}^{-1}\right)$, where $h_1, h_2, h_3, h_5$ were defined in Example~\ref{tabla n=15}, is the code $B_2(\alpha,5,11)$ of dimension $10$.
}\end{example}

\begin{example}\rm{
We also also show how to \textit{extend the dimension} of the codes in Example~\ref{tabla n=21}. We recall that $q=2$, $n=21$ and $\alpha$ satisfies that $\mathrm{min}(\alpha)=h_6$. In this case, we have the following BCH codes whose minimum distance and apparent distance coincide.

It is possible to modify the set $D_\alpha\left(\varphi_{\alpha,\overline{xg_2}}^{-1}\right)$ in three different ways. The biggest dimensional-extended code that we can obtain is $B_2(\alpha,2,0)$ of dimension $20$. In the case of $\left(\varphi_{\alpha,g_3}^{-1}\right)$, it determines two BCH codes. The first one is $B_2(\alpha,3,19)$ of dimension $15$ and the second one is $B_2(\alpha,3,12)$ of dimension $12$. The code $\left(\varphi_{\alpha,\overline{x^3g_4}}^{-1}\right)$ may be dimensional-extended to $B_2(\alpha,3,15)$ of dimension $12$, and $B_2(\alpha,3,1)$ of dimension $15$. The code $\left(\varphi_{\alpha,\overline{xg_5}}^{-1}\right)$ may be dimensional-extended to $B_2(\alpha,6,17)$ of dimension $11$. In the case of $\left(\varphi_{\alpha,\overline{x^5g_6}}^{-1}\right)$, we obtain $B_2(\alpha,6,0)$ of dimension $11$. Finally, $\left(\varphi_{\alpha,\overline{h_1h_3h_5h_6}}^{-1}\right)$ is the BCH code $B_2(\alpha,10,17)$ of dimension $10$.
}\end{example}

We finish with an example of a binary BCH code of length $33$ whose minimum distance equals the maximum of their BCH bounds. We have not found in the literature any binary BCH code satisfying that condition and having this length and dimension.

\begin{example}\rm{
 Set $q=2$, $n=33$ and $\alpha\in U_{33}$ such that $\mathrm{min}(\alpha)=x^{10}+x^7+x^5+x^3+1$ and $g=\mathrm{min}(\alpha)\mathrm{min}(\alpha^3)\mathrm{min}(\alpha^5)$. One may check that $g$ verifies the conditions of Theorem \ref{base true distance} with $k=0$ and $T=C_2(1)$; in fact $\varphi^{-1}_{\alpha,g}=x^{22}+x^{11}+1$. Hence, it determines $B_2(\alpha,3,31)$ of dimension $23$.
}\end{example}

\textbf{Acknowledgement}. The authors are grateful to the referees for many suggestions which contributed to improve this paper.

\medskip
% The data information below will be filled by AIMS editorial staff
Received xxxx 20xx; revised xxxx 20xx.
\medskip

\end{document}